\documentclass[11pt,letter]{article}
\usepackage[margin=1in]{geometry}
\usepackage{libertine}
\usepackage{xfrac}
\usepackage{bm}
\usepackage{datetime}
\usepackage{hhline}
\usepackage{multirow}
\usepackage{stmaryrd}
\usepackage{booktabs}
\usepackage{enumitem}
\usepackage{wrapfig}
\usepackage{subcaption}
\usepackage[T1]{fontenc}
\usepackage[utf8]{inputenc}
\usepackage{pgfplots}
\usepackage{amsmath, amssymb}
\usepackage{amsthm}
\usepackage{color}
\usepackage{cases}
\usepackage{xparse}
\usepackage{xargs}
\usepackage{appendix}
\usepackage[ruled,vlined,noend,linesnumbered]{algorithm2e}
\usepackage{algpseudocode}
\usepackage{verbatim, xspace}
\usepackage{tikz}
\usepackage{mathtools}
\usepackage{mleftright}
\usepackage{thm-restate}
\usepackage[framemethod=TikZ]{mdframed}
\usepackage[sort]{cite}

\SetCommentSty{FontInPseudocodeComments}
\usepackage{hyperref}
\usepackage{pgfplots}

\usepackage{xcolor}
\usepackage[noabbrev,capitalize]{cleveref}

\newcommand{\citet}[1]{\cite{#1}}

\newtheorem{theorem}{Theorem}

\newtheorem{lemma}[theorem]{Lemma}

\newtheorem{claim}[theorem]{Claim}

\newtheorem*{maingoal*}{Main Question}

\numberwithin{theorem}{section}
\numberwithin{lemma}{section}
\numberwithin{claim}{section}
\numberwithin{corollary}{section}
\numberwithin{definition}{section}
\numberwithin{observation}{section}
\numberwithin{proposition}{section}

\crefname{claim}{Claim}{Claims}
\Crefname{claim}{Claim}{Claims}

\newenvironment{claimproof}[1][\proofname]
{\proof[#1]}
{\endproof}

\newcommand{\norm}[1]{\mleft\lVert#1\mright\rVert}

\newcommand{\ceil}[1]{\mleft\lceil #1 \mright\rceil}

\newcommand{\eps}{\varepsilon}

\newcommand{\Oh}{\mathcal{O}}

\newcommand{\Z}{\mathbb{Z}}

\newcommand{\poly}{\mathrm{poly}}

\newcommand{\Mod}[1]{\ (\mathrm{mod}\ #1)}

\newcommand{\iv}[1]{\llbracket #1 \rrbracket}

\newcommand{\veca}{\normalfont{\textbf{a}}}
\newcommand{\vecb}{\normalfont{\textbf{b}}}
\newcommand{\vect}{\normalfont{\textbf{t}}}

\newcommand{\vecv}{\normalfont{\textbf{v}}}
\newcommand{\vecx}{\normalfont{\textbf{x}}}
\newcommand{\vecy}{\normalfont{\textbf{y}}}
\newcommand{\vecz}{\normalfont{\textbf{z}}}

\newcommand{\vecs}{\normalfont{\textbf{s}}}
\newcommand{\vecr}{\normalfont{\textbf{r}}}

\newcommand{\veczero}{\normalfont{\textbf{0}}}
\newcommand{\ILP}{\normalfont{\text{ILP}}}

\renewcommand{\leq}{\leqslant}
\renewcommand{\geq}{\geqslant}
\renewcommand{\le}{\leqslant}
\renewcommand{\ge}{\geqslant}

\renewcommand{\epsilon}{\varepsilon}

\title{A Simple Algorithm for Counting Solutions to Integer Programs}
\date{}
\author{
    Daniel Dadush\footnote{Centrum Wiskunde \& Informatica (CWI), Amsterdam, The Netherlands, \textsf{dadush@cwi.nl}.}
    \and 
  Anita D\"urr\footnote{ETH Z\"urich, Z\"urich, Switzerland, \textsf{anita.duerr@inf.ethz.ch}.}
    \and
    Thekla Hamm\footnote{Eindhoven University of Technology, Eindhoven, The Netherlands, \textsf{t.l.s.hamm@tue.nl}.}
    \and
    Jonas Schmidt\footnote{Bocconi University, Milan, Italy, \textsf{jonas.schmidt2@phd.unibocconi.it}.}
    \and
    Karol W\k{e}grzycki\footnote{Max Planck Institute for Informatics,
        Saarbr\"ucken, Germany, \textsf{kwegrzyc@mpi-inf.mpg.de}. Supported by the Deutsche Forschungsgemeinschaft (DFG, German Research Foundation) grant number 559177164. }
}

\begin{document}

\maketitle

\thispagestyle{empty}
\begin{abstract}
    We consider the problem of counting solutions to integer linear programming
    problems of the form $\{ \vecx \in \mathbb{Z}^n_{\ge 0} : A\,\vecx = \vecb\}$ where
    $A \in \{-\Delta,\ldots,\Delta\}^{m \times n}$ and $\vecb \in \mathbb{Z}^m$.
    We present an algorithm that counts solutions with $(n\Delta)^{m + o(1)}\! \cdot \log\norm{\vecb}$ arithmetic operations. This simplifies the algorithm of Gribanov and Malyshev, who
    considered the problem for the more general setting in which $A$ is $\Delta$-modular.
\end{abstract}

\vfill
\paragraph{Acknowledgements}
We thank the organizers of the workshop \textit{Phase Transitions in Combinatorics, Algorithms and Probability} held at the CWI, Amsterdam in May 2026 for providing a fruitful research environment in which this work was initiated.

\clearpage
\setcounter{page}{1}

\section{Introduction}

Integer linear programming is a fundamental
problem in discrete optimization and theoretical computer science.
In this paper, we study the problem of computing the cardinality of
\begin{equation*}
    \ILP(A,\vecb) \coloneq \mleft\{ \vecx \in \mathbb{Z}_{\ge 0}^n : A\vecx = \vecb \mright\},
\end{equation*}
where $A \in \{-\Delta,\ldots,\Delta\}^{m \times n}$ and $\vecb \in
\mathbb{Z}^m$.

The corresponding feasibility problem asks whether $\ILP(A,\vecb)$ is nonempty.
Building on the seminal work of Lenstra and Kannan~\cite{lenstra1983integer,kannan1983improved}, Reis and Rothvoss~\cite{reis2023subspace} gave a randomized algorithm deciding feasibility in time
$
(\log(2n))^{\Oh(n)} \cdot \poly(|I|),
$
where $|I|$ denotes the input size. 
Note, however, that these algorithms only determine whether a solution exists and do not compute the number $|\ILP(A, \vecb)|$ of solutions. For the counting problem, the best known general-purpose method is Barvinok's algorithm~\cite{barvinok1993polynomial,barvinok2003short}, which computes $|\ILP(A,\vecb)|$ in time $|I|^{\Oh(n\log n)}$ (see~\cite{dyer1997barvinok}).

These algorithms are exponential when $n$ is large. Gribanov and Malyshev~\cite{gribanov2022faster} considered the regime in which $A$ is $D$-modular, i.e., $D(A) \leq D$, where, using \(A_{R,C}\) to denote the restriction of \(A\) to a subset \(R\) of its rows and a subset \(C\) of its columns,
\[
D(A)\coloneqq \max\mleft\{|\det(A_{R,C})| \text{ such that } |R|=|C|=\operatorname{rank}(A)\mright\}.
\]
Gribanov and Malyshev~\cite{gribanov2022faster} presented an algorithm computing $|\ILP(A,\vecb)|$ using
$\Oh\big(n^4 \nu^2  D^4 \log{D}\big)$ arithmetic operations,
where $\nu = \nu(A)$ is the maximum possible number of vertices in
the polytope defined by $A$ as $\vecb$ varies (see~\cite[Definition 2]{gribanov2022faster}). Since $D(A)\le \Oh(\sqrt{m}\Delta)^m$
and $\nu(A)\le \Oh(n^m)$, their bound simplifies to $\Oh(\sqrt{mn}\Delta)^{4m}\cdot \poly(|I|)$ arithmetic operations.
More recently, Gribanov, Shumilov, Malyshev, and Zolotykh~\cite{gribanov2022sparse} (see also~\cite{gribanov2023representation})
gave an improved bound of $\Oh(\nu^2 d^4 D^3)$ arithmetic operations, where $d$ is the polytope dimension.
Since $d\le n$, the same estimates for $\nu(A)$ and $D(A)$ give
\[
\Oh\mleft(n^{2m+4} m^{3m/2} \Delta^{3m}\mright).
\]
in the parameters used here.
In this work, we give a simpler algorithm with improved dependence on $n$, $m$, and $\Delta$.

\begin{theorem}\label{thm:main}
    There exists an algorithm that computes $|\ILP(A,\vecb)|$ with $\Oh(n\Delta)^{m}  \log(n\Delta)\log\bigl(n\Delta\norm{\vecb}_\infty\bigr)$ arithmetic operations\footnote{To simplify the presentation, we bound the running time always in the number of arithmetic operations. In the standard word RAM model, the running time equals the stated bound increased by a factor of $O(n\log(\Delta\norm{\vecb}))$ (see the discussion in \Cref{sec:general}).}.
\end{theorem}

Notice that in the packing case (i.e., when every entry of $A$ is
non-negative), after handling zero columns separately, the value
$|\ILP(A,\vecb)|$ can be computed using
$\Oh\bigl(n\norm{\vecb}_\infty^m\bigr)$ arithmetic operations by dynamic programming, but
even in this special case it is not clear how to avoid polynomial dependence on
$\norm{\vecb}_\infty$. For general $A$, the feasibility problem under this
parameterization was first studied by
Papadimitriou~\cite{papadimitriou1981complexity}. Eisenbrand
and Weismantel~\cite{eisenbrand2019proximity} proposed an $\Oh(m
\Delta)^{m} \cdot n\norm{\vecb}_1$-time algorithm. Jansen and
Rohwedder~\cite{jansen2023integer} later improved this to a bound
of $\Oh(m\Delta)^m \log(\Delta) \log(\Delta +
\norm{\vecb}_\infty) + \Oh(nm)$. All of these methods rely on proximity results
guaranteeing the existence of an integral solution close to the LP relaxation.
This approach is challenging to apply in the counting setting.

\paragraph{Organization and Exposition.}
Our paper follows a slightly non-standard structure.
We start by presenting the simpler packing case of \Cref{thm:main} in \Cref{sec:packing}.
The algorithm is presented in the language that we discovered it in, namely using a generating function of $\ILP(A,\vecb)$ and polynomial multiplications and manipulations of polynomials.
This language turns out to be convenient for expressing the operations that our algorithm performs.
We find this surprising, since our algorithm does not rely on any specific deeper theory used in other counting algorithms such as Barvinok's algorithm~\cite{barvinok1993polynomial,barvinok2003short}. (In fact, our algorithm even works with a different generating function.)

While allowing for an efficient presentation and interesting perspective, we find it less intuitive to grasp which operations the algorithm in \Cref{sec:packing} is actually performing and therefore harder to generalize (e.g., to the non-packing case). 
Thus, we believe that it adds value to present the generalized algorithm and full proof of \Cref{thm:main} in a more traditional, dynamic-programming-based way in \Cref{sec:general}.

Based on this difference, readers may choose whether to read \Cref{sec:packing} or \Cref{sec:general}, as the main ideas are identical.
The main difference is the need for the Steinitz lemma (see,~\cite{eisenbrand2019proximity}).

\section{Packing ILP}\label{sec:packing}

First, we focus on the special case of packing ILP, i.e., when every entry of $A$ is non-negative.  
Our algorithm relies on manipulating the generating polynomial of $\ILP(A,\vecb)$. In order to multiply two polynomials, we repeatedly use the following standard slight generalization of the Fast Fourier Transform (FFT) algorithm to multivariate polynomials (see, e.g.,~\cite[Chapter~30]{cormenetal} together with Exercise 30-3).

\begin{lemma}[Fast Fourier Transform,~\protect{\cite[Chapter~30]{cormenetal}}]
    \label{lem:fft}
    There is an algorithm that, given two $m$-variate polynomials $P(\vecz)$ and $Q(\vecz)$ with degree at most $\Gamma$ per variable and represented by their monomials, computes the monomials of $P(\vecz) \cdot Q(\vecz)$ using $O(\Gamma)^m\cdot \log \Gamma$ operations.
\end{lemma}

The goal of this section is to prove the following statement.

\begin{lemma}\label{thm:packing}
    There exists an algorithm that computes $|\ILP(A,\vecb)|$ with $\Oh(n\Delta)^{m}  \log(n\Delta) \log{(n\Delta\norm{\vecb}_\infty)}$ arithmetic operations when $A \in \{0,\ldots,\Delta\}^{m \times n}$.
\end{lemma}

In this case, we use the characteristic function 
approach which was already used to simplify algorithms for the related Subset Sum
problem~\cite{jin2019simple}. 
\begin{proof}[Proof of~\cref{thm:packing}]
    If some coordinate of $\vecb$ is negative, then clearly
    $\ILP(A,\vecb)=\emptyset$. If $A$ contains any zero columns, we delete them and return
    $\infty$ if the algorithm detects at least one solution.
    Hence, for the remainder of the proof, we may assume that every column of
    $A$ is nonzero and $\vecb\in\mathbb{Z}_{\geq 0}^m$.

    The reader is invited to assume $m=1$ on the first read. For arbitrary $m \ge 1$, we use $\vecz \in \mathbb{Z}^m$ to denote $m$ variables $(z_1,\ldots,z_m)$.
    Let $F(\vecz) = F(z_1,\ldots,z_m)$ be a generating function of $\ILP(A,\vecb)$, more precisely:
    \begin{align*}
        F(\vecz) &\coloneqq \sum_{\vecb \in \mathbb{Z}^m} |\ILP(A,\vecb)| \cdot \vecz^{\vecb},
    \end{align*}
    where $\vecz^{\vecb} = z_1^{b_1} \cdots z_m^{b_m}$. Our goal is to compute the coefficient of $\vecz^{\vecb}$ in $F(\vecz)$ denoted by $\iv{\vecz^{\vecb}}\, F(\vecz)$.
    Towards this goal, 
    let $\veca_i \in \{0,\ldots,\Delta\}^m$ be the $i$-th column of $A$, and define the polynomial $$Q(\vecz) \coloneqq \prod_{i=1}^n \mleft(1+\vecz^{\veca_i}\mright).$$ 
    Let $\Gamma \coloneqq n\Delta$ be an upper bound on the degree of each variable $z_i$ in $Q(\vecz)$. 
    We observe the following.
    
     \begin{claim}\label{clm:correctness}
            $F(\vecz) = \prod_{k \geq 0} Q(\vecz^{2^k}) = Q(\vecz) \cdot F(\vecz^2)$.
    \end{claim}
    \begin{claimproof}
    By the definition of $F(\vecz)$ and considering that $x_i = k_i$ is a solution to $A\vecx = \sum_i \veca x_i$ we have:
    \begin{align*}
        F(\vecz) &= \prod_{i=1}^n \sum_{k_i \ge 0} \mleft(\vecz^{\veca_{i}}\mright)^k_i 
    \end{align*}
    Next, we use the following identity that follows by the uniqueness of the binary expansion:
    \begin{align*}
        \sum_{k \ge 0} x^k &= \prod_{k\ge 0} \mleft(1+x^{2^k}\mright).\nonumber
    \end{align*}
    This means that $F(\vecz)$ can be represented in the following form
    \begin{align*}
        F(\vecz) &= \prod_{i=1}^n \prod_{k\ge 0} \mleft(1+(\vecz^{\veca_i})^{2^k}\mright)
                = \prod_{k\ge 0} \prod_{i=1}^n \mleft(1+\vecz^{\veca_i \cdot 2^k}\mright) = \prod_{k \geq 0} Q(\vecz^{2^k}).
    \end{align*}
    By taking out one factor of $Q(\vecz)$ from this product, the final equality immediately follows:
        \[F(\vecz) = Q(\vecz) \cdot F(\vecz^2).\qedhere\]
    \end{claimproof}

From the above claim, we see that the coefficients of $F(\vecz)$ are functions of the coefficients of $Q(\vecz)$. In particular, $Q(\vecz)$ contains all the information about $F(\vecz)$. 

We now present our recursive algorithm for computing coefficients of $F(\vecz)$. Due to the nature of the recursive subproblems, the algorithm solves the more general problem of computing a target coefficient of the product of $F(\vecz)$ with another polynomial $P(\vecz)$, where the original problem corresponds to $P(\vecz)=1$. The input to the algorithm is the target coefficient $\vecb \in \Z^m_{\geq 0}$, and the list of coefficients of $P(\vecz)$ and $Q(\vecz)$ (recalling that $Q(\vecz)$ represents $F(\vecz)$ implicitly).
    
\begin{algorithm}[ht!]
\caption{Procedure \textsc{Coeff}: recursive coefficient extraction for \cref{clm:recursion}}
\label{alg:recursion}
\SetArgSty{upshape}
\SetKwFunction{Coeff}{\textnormal{\textsc{Coeff}}}
\KwIn{Polynomials $P(\vecz)$ and $Q(\vecz)$, and a vector $\vecb\in\mathbb{Z}_{\geq 0}^m$}
\KwOut{The coefficient $\iv{\vecz^{\vecb}}\,P(\vecz)\cdot F(\vecz)$}
\If{$\vecb=\mathbf{0}$}{
  \Return{$P(\mathbf{0})$}
}
Compute the monomials of $P(\vecz) \cdot Q(\vecz)$ via FFT (\Cref{lem:fft})\\
Group the monomials of $P(\vecz) \cdot Q(\vecz)$ by parity of exponents to compute the decomposition $P(\vecz)\cdot Q(\vecz)=\sum_{\vecs\in\{0,1\}^m}\vecz^{\vecs}\,V_{\vecs}(\vecz^2)$\\
$\vecr\gets(b_1\bmod 2,\ldots,b_m\bmod 2)$\\
\Return{$\Coeff(V_{\vecr}(\vecz),Q(\vecz),(\vecb-\vecr)/2)$}\\
\end{algorithm}
    \begin{claim}\label{clm:recursion}
        Let $P(\vecz)$ be an $m$-variate polynomial where each variable has degree at most $\Gamma$. Then, given $\vecb \in \mathbb{Z}_{\geq 0}^m$ and the monomials of $P(\vecz)$ and $Q(\vecz)$, we can compute:
        \begin{align*}
            \iv{\vecz^{\vecb}} \,P(\vecz) \cdot F(\vecz)
        \end{align*}
        with $\Oh(\Gamma)^m \cdot \log(\Gamma) \cdot \log\norm{\vecb}_\infty$ arithmetic operations.
    \end{claim}
    \begin{claimproof}
        By grouping monomials of $P(\vecz)\cdot Q(\vecz)$ we can represent it by uniquely determined polynomials $V_{\vecs}(\vecz)$ as follows:
        \begin{align}\label{eq:representation}
            P(\vecz) \cdot Q(\vecz) = \sum_{\vecs \in \{0,1\}^m} \vecz^{\vecs} \cdot V_{\vecs}(\vecz^2).
        \end{align}
        As $P(\vecz) \cdot Q(\vecz)$ is a polynomial of degree at most $2\Gamma$ per variable, each polynomial $V_{\vecs}(\vecz)$ for $\vecs \in \{0,1\}^m$ has degree at most $\Gamma$ per variable. Next, we use~\cref{clm:correctness} to notice that:
        \begin{align*}
            P(\vecz) \cdot F(\vecz) & = P(\vecz) \cdot Q(\vecz) \cdot F(\vecz^2) 
            = \sum_{\vecs \in \{0,1\}^m } \vecz^{\vecs} \cdot V_{\vecs}(\vecz^2) \cdot F(\vecz^2).
\intertext{Let $\vecr \in \{0,1\}^m$ be such that $b_i \equiv r_i \,(\text{mod } 2) \text{ for all }  1 \le i \le m$. Then:
}
            \iv{\vecz^{\vecb}}\, P(\vecz) \cdot F(\vecz) &= 
            \iv{\vecz^{\vecb}}\, \sum\nolimits_{\vecs \in \{0,1\}^m} \vecz^{\vecs} \cdot V_{\vecs}(\vecz^2)\cdot F(\vecz^2) \\
            &=\sum\nolimits_{\vecs \in \{0,1\}^m} \iv{\vecz^{\vecb - \vecs}}\, V_{\vecs}(\vecz^2)\cdot F(\vecz^2) \\
            &=\iv{\vecz^{\vecb-\vecr}}\,\, V_{\vecr}(\vecz^2)\cdot F(\vecz^2) \tag{Only even $\vecb - \vecs$ contribute}\\
            &=\iv{\vecz^{(\vecb-\vecr)/2}}\, V_{\vecr}(\vecz) \cdot F(\vecz).
        \end{align*}
        Observe that $V_{\vecr}(\vecz)$ is an $m$-variate polynomial where each variable has degree at most $\Gamma$. So,
        we can compute the representation in~\eqref{eq:representation} of $V_{\vecr}(\vecz)$ using $\Oh((\Gamma+1)^m \log((\Gamma + 1)^m)) \le \Oh(\Gamma)^m \cdot \log{\Gamma}$ arithmetic operations with FFT (\Cref{lem:fft}) and recursively compute $\iv{\vecz^{(\vecb-\vecr)/2}}\, V_{\vecr}(\vecz) F(\vecz)$.

        We summarize the procedure in \cref{alg:recursion}.
        In each step of the recursion $\norm{\vecb}_\infty$ halves, so the depth
        of the recursion is $\Oh(\log \norm{\vecb}_\infty)$. In the base case, we
        conclude that for $\vecb = \mathbf{0}$ we have $\iv{\vecz^{\mathbf{0}}}
        \, P(\vecz) \cdot F(\vecz)= P(\mathbf{0})$. Indeed, since $A$ is
        non-negative and has no zero column, $A\vecx=\mathbf{0}$ implies
        $\vecx=\mathbf{0}$.
    \end{claimproof}

    Now notice that we can compute the monomials of $Q(\vecz)$ as follows.
    \begin{claim}\label{clm:Qz}
        The monomials of $Q(\vecz)$ can be computed in $\Oh(\Gamma)^m \cdot \log^2\Gamma$ arithmetic operations.
    \end{claim}
    \begin{claimproof}
        We recursively split the set of columns $\{1, \dots, n\}$ into two
        halves, compute the product of the terms in each half, and then multiply
        the two resulting polynomials.
        
        This yields a binary tree of multiplications with depth $\Oh(\log n)$.
        For each level $\lambda \in \{0,\ldots,\ceil{\log{n}}\}$ of the tree, the
        degree per variable of the resulting polynomials is at most $2^\lambda \cdot \Delta$, and
        the number of polynomials at level $\lambda$ is at most
        $2^{\ceil{\log{n}}-\lambda}$. Using FFT~(\Cref{lem:fft}), each multiplication takes
        $\Oh(2^{\lambda} \Delta)^m \cdot \log(n \Delta)$ operations.
        Summing over all levels, the total cost is bounded by 
        \begin{equation*}
            \sum_{\lambda=0}^{\lceil\log n\rceil} 2^{\ceil{\log{n}}-\lambda}
            \cdot \Oh(2^{\lambda}  \Delta)^m \cdot \log(n \Delta) \le
            \Oh(n\Delta)^m \cdot \log^2(n\Delta).\qedhere
        \end{equation*}
    \end{claimproof}

    Hence, to compute $\iv{\vecz^{\vecb}}\, F(\vecz)$, we first compute $Q(\vecz)$ using \cref{clm:Qz} and apply~\cref{clm:recursion} with $P(\vecz) = 1$. In total the number of arithmetic operations is 
    \begin{equation*}
        \Oh(\Gamma)^m \cdot \log \Gamma \cdot \log \norm{\vecb}_\infty +  \Oh(\Gamma)^m \cdot \log^2\Gamma \leq \Oh(n\Delta)^{m} \cdot \log(n\Delta) (\log\norm{\vecb}_\infty + \log(n\Delta)).\qedhere
    \end{equation*}
\end{proof}

Notice that we used the assumption that $A$ is non-negative only in
the base case of the recursion, where we trivially have $\iv{\vecz^{\mathbf{0}}}\,
F(\vecz) = 1$. This is not necessarily $1$ (or even finite!) in the case when
$A$ has negative entries. In the next section, we show how to handle this using
the Steinitz lemma.

\section{Proof of~\cref{thm:main}}\label{sec:general}

We will require the following bound on the size of solutions to an integer program. At a high level, the bound will be used to prove that a feasible integer program always has short solutions, and that long solutions only exist if there are infinitely many solutions.

\begin{lemma}\label{cor:steinitz-apply}
    Let $\vect \in \mathbb{Z}^m$ such that there is a solution $\vecx \in\mathbb{Z}^n_{\ge 0}$  to $A\vecx = \vect$ with \[\norm{\vecx}_1 > \Gamma_{\vect} \coloneqq (\norm{\vect}_\infty + 4m\Delta + 1)^m.\] 
    Then, the following holds:
    \begin{enumerate}
        \item There is a nonzero solution $\vecx' \in \mathbb{Z}^n_{\ge 0}$ to $A\vecx' = \vect$ with $\norm{\vecx'}_1 \le \Gamma_{\vect}$.
        \item There is a nonzero $\vecy \in \mathbb{Z}^n_{\ge 0}$ with $A\vecy = \veczero$ with $\norm{\vecy}_1 \leq \Gamma_{\veczero}$.
    \end{enumerate}
\end{lemma}

The bound can be derived from standard results in integer programming, in particular using Graver bases and conformal decompositions. To make the paper more self-contained, we give a direct proof using the Steinitz lemma, which was introduced as tool for proving related size bounds by Eisenbrand and Weismantel~\cite[Section~2]{eisenbrand2019proximity}. Intuitively, it states that a sequence of vectors that sum to a target $\vecb$ can be rearranged so that their partial sums remain in a tube of radius $2m\Delta$ around the line segment between $\veczero$ and $\vecb$.

\begin{lemma}[Steinitz lemma, tube version~\protect{\cite[Section 2]{eisenbrand2019proximity}}]\label{lem:steinitz}
    Let $\vect, \vecv_1, \dots, \vecv_n \in \mathbb{R}^m$ such that $\norm{\vecv_i}_\infty \le \Delta$ for all $i$ and $\sum_{i=1}^n\vecv_i = \vect$.
    Then there exists a permutation $\pi$ such that for all partial sums $k \in \{1, \dots, n\}$:
    \begin{equation}
    \label{eq:tube}
        \sum_{i=1}^k \vecv_{\pi(i)} \in \big\{\vecx : \exists j \in \{1,\dots,n\}, \norm{\vecx - (j/n)\vect}_\infty \le 2m\Delta \big\}.
    \end{equation}
\end{lemma}

\begin{proof}[Proof of \Cref{cor:steinitz-apply}]
    \begin{sloppypar}
    Examine the sequence of vectors derived by repeating the columns of $A$ the number of times they appear in $\vecx$.
    By the Steinitz lemma (\Cref{lem:steinitz}), we can construct a rearranged sequence $\vecv_1, \dots, \vecv_{\norm{\vecx}_1}$ that sums to $\vect$ and whose prefix sums never leave the tube~\eqref{eq:tube} around the line segment from $\veczero$ to $\vect$.
    Since the box $t/2 + [\pm(\norm{\vect}_\infty/2 + 2m \Delta)]^m$ contains the tube and has $\Gamma_{\vect}$ integer points, the number of integer points in the tube is at most $\Gamma_{\vect}$.

    Given that $\norm{\vecx}_1 > \Gamma_{\vect}$, by pigeon-hole we can find two distinct prefixes $k,k'$ of the Steinitz sequence with equal partial sums and $1 \leq k < k' \leq \Gamma_{\vect}+1$. The difference of their partial sums then satisfies $\veczero = \sum_{i=k}^{k'} \vecv_i := A \vecy$, where $0 \leq \vecy \leq \vecx$ and $0 < \norm{\vecy}_1 = k'-k \leq \Gamma_{\vect}$.

    To prove property~1, examine $\vecx' := \vecx-\vecy$. By construction, $\vecx' \neq \veczero$ is a non-negative integer solution to $A\vecz = \vect$ with strictly smaller $\norm{\cdot}_1$-norm than $\vecx$. Applying this procedure iteratively, replacing $\vecx$ with the shorter $\vecx'$ at each step, we may continue until $\norm{\vecx} \leq \Gamma_{\vect}$.

    To prove property~2, we invoke property~1 of the Lemma when applied against the initial non-zero $\vecy$ found in the second paragraph of the proof.
    \qedhere
    \end{sloppypar}
\end{proof}

To prove \Cref{thm:main}, we use a dynamic program that should be thought of as mostly equivalent to the polynomial approach in \Cref{sec:packing}, except running for more iterations. 
Rohwedder and W{\k{e}}grzycki~\cite{rohwedder2025fine} apply a similar dynamic programming technique for optimizing an ILP where variables are bounded. Their algorithm relies on proximity, which we avoid by using \Cref{cor:steinitz-apply} directly.

Unlike in the packing case in \Cref{thm:packing}, we need to be slightly more careful about checking if the answer is finite: It requires us to check feasibility and the existence of a \emph{cycle}, that is, property 2 of \Cref{cor:steinitz-apply}.
Both of these problems have standard solutions: A cycle can be found in polynomial time by solving an LP over the reals, and there are fast algorithms for checking feasibility~\cite{eisenbrand2019proximity}.
To achieve a precise value of $m$ in the exponent, we take a different route and do both of these checks with our dynamic programming techniques for counting (see \Cref{cl:steinitz}).

\begin{proof}[Proof of \Cref{thm:main}]
    Assume without loss of generality $m > 1$; the case $m = 1$ is either covered by the packing case (\Cref{thm:packing}) or $|\ILP(A,\vecb)| \in \{0, \infty\}$. 
    Moreover, $|\ILP(A,\vecb)| = \infty$ if and only if the greatest common divisor of the values of $A$ divides the right-hand side, which we simply check in the beginning.

    In the following, we represent a multiset by storing the distinct set elements along with the number of their occurrences in binary. 

    Consider the following multiset for $i \ge 0$ and 
    $\vect \in \mathbb{Z}^m$
    \begin{equation*}
        C_i(\vect) \coloneqq \mleft\{ A\vecx : \vecx \in \{0, \dots, 2^i-1\}^n \text{ and } A\vecx \equiv \vect \Mod{2^i} \mright\}. 
    \end{equation*}
    We observe that to compute $|\ILP(A, \vecb)|$, it suffices to compute the multisets $C_{i^*}(\vecb)$ and $C_{i^*}(\veczero)$, where
    $i^* \coloneqq \lceil \log_2 \Gamma_{\vecb} \rceil + 1$ and
    $\Gamma_{\vecb}$ is as defined in \cref{cor:steinitz-apply}. In particular,
    $i^*$ is a positive integer and $2^{i^*} > \Gamma_{\vecb}$.
    \begin{claim}\label{cl:steinitz}
        The following properties hold:
        \begin{enumerate}
            \item If $\vecb \notin C_{i^*}(\vecb)$, then $|\ILP(A,\vecb)| = 0$.
            \item \label{it:2} If $\vecb \in C_{i^*}(\vecb)$ and  $\veczero$ occurs at least twice in $C_{i^*}(\veczero)$, then $|\ILP(A,\vecb)| = \infty$.
            \item In all other cases, $\vecb$ occurs $|\ILP(A,\vecb)|$ times in $C_{i^*}(\vecb)$.
        \end{enumerate}
    \end{claim}
    \begin{claimproof}
        By \Cref{cor:steinitz-apply}, if $|\ILP(A,\vecb)| > 0$ then there must be a solution $\vecx \in \ILP(A, \vecb)$ with $\norm{\vecx}_1 \le \Gamma_{\vecb}$.
        This $\vecx$ must be counted in $C_{i^*}(\vecb)$, proving property~1.
        
        For property~2, notice that 
        $\veczero$ always occurs in $C_{i^*}(\veczero)$ because  $A\veczero=\veczero$. 
        Thus, a second occurrence of $\veczero$ yields a nonzero vector 
        $\vecy \in \mathbb Z_{\ge 0}^n$ with $A\vecy=\veczero$.
        If $\vecb \in C_{i^*}(\vecb)$ then there exists $\vecx$ with $A\vecx = \vecb$, and thus $\ILP(A,\vecb)$ contains infinitely many vectors of the form $\{\vecx + k\vecy\}_{k \in \mathbb{Z}_{\ge 0}}$. 

        To prove property~3, suppose there is a solution $\vecx' \in \ILP(A, \vecb)$ such that $A \vecx'$ is not represented in $C_{i^*}(\vecb)$. This means that \[\norm{\vecx'}_1 \ge \norm{\vecx'}_\infty \ge 2^{i^*} > \Gamma_{\vecb}\] and by \Cref{cor:steinitz-apply} there must be a nonzero $\vecy \in \mathbb{Z}^n_{\ge 0}$ with $A\vecy = \veczero$ and $\norm{\vecy}_1 \le \Gamma_{\veczero} \le \Gamma_{\vecb}$.
        Hence, both $A \veczero$ and $A \vecy$ are contained in $C_{i^*}(\veczero)$, i.e., $\veczero$ occurs at least twice in $C_{i^*}(\veczero)$. So the assumptions of property~2 are satisfied, a contradiction.
    \end{claimproof}

    The goal of the algorithm is therefore to compute the multisets $C_{i^*}(\vecb)$ and $C_{i^*}(\veczero)$.
    Note that $C_0(\vect) = \{\veczero\}$ for any $\vect \in \Z^m$.
    We show how to efficiently compute $C_{i+1}(\vect)$ from $C_i(\vect)$. 
    Define the following multiset
    \begin{equation*}
        B \coloneqq \mleft\{A\vecx : \vecx \in \{0,1\}^n \mright\}.
    \end{equation*}
    
    We establish the following recurrence.
    \begin{claim}
        \label{cl:transition}
        \[
        C_{i+1}(\vect) = \big(C_i(\vect) + 2^iB\big) \, \cap \, \mleft\{\vecr \in \mathbb{Z}^m : \vecr \equiv \vect \Mod{2^{i+1}}\mright\}.
        \]
    \end{claim}
    \begin{claimproof}
        The right-hand side clearly enforces the residue condition and is hence included in the left-hand side. For the other inclusion, consider $\vect' \in C_{i+1}(\vect)$ with $A\vecx = \vect'$. Clearly, $\vect' \equiv \vect \Mod{2^{i+1}}$. Furthermore, define 
        \begin{equation*}
            \vecx' \coloneqq \vecx \Mod{2^i} \quad \text{and} \quad \vecx'' \coloneqq \vecx - \vecx'.
        \end{equation*}
        Since $\vecx'' \in \{0,2^i\}^n$, we have $A\vecx'' \in 2^i B$.
        Moreover
        \begin{equation*}
            \vect \equiv \vect' \equiv A\vecx \equiv A\vecx' + A\vecx'' \Mod{2^{i+1}},
        \end{equation*}
        and since $A\vecx'' \equiv 0 \Mod{2^i}$,  we have $A\vecx' \in C_i(\vect)$. This shows that $\vect' \in (C_i(\vect) + 2^i B)$.
        
        Additionally, observe that this mapping also defines a bijection, hence the number of occurrences of each element in the multisets match as well.
    \end{claimproof}
    
    \begin{algorithm}[ht!]
        \caption{Our dynamic programming counting algorithm for \cref{thm:main}}
        \label{alg:main}
        \SetArgSty{upshape}
        \KwIn{A matrix $A \in \{-\Delta, \ldots, \Delta\}^{m \times n}$ and a vector $\vecb\in\mathbb{Z}^m$}
        \KwOut{The count $|\ILP(A,\vecb)|$}
        Compute $B = \{A\vecx : \vecx \in \{0,1\}^n\}$\\
        \For{$\vect \in \{\veczero, \vecb\}$}{
        Initialize $C_0(\vect) = \{\veczero\}$\\
        \For{$i = 0, \ldots, i^*-1 = \lceil\log_2\Gamma_{\vecb}\rceil$}{
            Compute the multiset $C_{i+1}(\vect) = (C_i(\vect) + 2^iB) \, \cap \, \{\vecr \in \mathbb{Z}^m : \vecr \equiv \vect \Mod{2^{i+1}}\}$
        } 
        }
        \If{$\vecb \notin C_{i^*}(\vecb)$}{\Return{0}}
        \ElseIf{$\vecb \in C_{i^*}(\vecb)$ and $\veczero$ occurs at least twice in $C_{i^*}(\veczero)$}{\Return{$\infty$}}
        \Else{\Return{the number of occurrences of $\vecb$ in $C_{i^*}(\vecb)$}}
    \end{algorithm}

    Our algorithm works as follows. 
    First, compute $B$ and initialize $C_0(\vecb) = C_0(\veczero) = \{\veczero\}$. Then, iteratively compute $C_{i}(\vecb)$ and $C_{i}(\veczero)$ for $i = 1, \dots, i^*$ using the previously computed transition (\Cref{cl:transition}). 
    Using \cref{cl:steinitz} and the number of occurrences of $\vecb$ in $C_{i^*}(\vecb)$ and of $\veczero$ in $C_{i^*}(\veczero)$, we can now deduce the correct value of $|\ILP(A, \vecb)|$, including whether it is infinite (see also \Cref{alg:main}, the case $m=1$ is omitted there since it only influences $\log$ factors).
    
    Finally, we prove the bound on the number of operations.
    \begin{claim}\label{cl:time}
        $B$ can be computed using $\Oh(n\Delta)^m \cdot \log(n\Delta)$ operations. Given $C_i(\vect)$ and $B$, $C_{i+1}(\vect)$ can also be computed using $\Oh(n\Delta)^m\cdot \log(n\Delta)$ operations.
    \end{claim}
    \begin{claimproof}
        The computation of $B$ proceeds similarly to \Cref{clm:Qz}.
        Notice that any sum of at most $k$ column vectors of $A$ stays in the set $\{-k\Delta, \dots, k\Delta\}^m$.
        We use recursion to compute $B$: Divide the set of columns $\{1, \dots, n\}$ into two roughly equal-sized parts and recursively compute $B$ including the counts for each vector for $A$ restricted to both parts.
        Finally, combine the two parts with FFT\footnote{A well-known observation is that this operation is equivalent to multiplying two polynomials by treating the vector coordinates as exponents and the multiplicities as factors in front of the corresponding monomials. Hence, \Cref{lem:fft} applies.}~(\Cref{lem:fft}) to obtain the whole set $B$.
        On recursion level $\lambda \in \{0, \dots, \lceil\log n\rceil\}$, there are at most $2^{\lceil\log n\rceil - \lambda}$ FFT instances with vectors in $\{-2^{\lambda}\Delta, \dots, 2^{\lambda}\Delta\}^m$ each.
        Since we assume $m > 1$, the cost per recursion level decreases geometrically, and we can bound the total number of operations by
        \begin{equation*}
            \sum_{\lambda=0}^{\lceil\log n\rceil} 2^{\lceil\log n\rceil - \lambda} \Oh(2^\lambda\Delta)^m \cdot m\log (2^\lambda \Delta) \le \Oh(n\Delta)^m \log(n\Delta) \sum_{\lambda=0}^{\lceil\log n\rceil}  2^{-\lambda(m-1)}  \le \Oh(n\Delta)^m \log(n\Delta).
        \end{equation*}

        To compute $C_{i+1}(\vect)$, we use \Cref{cl:transition}. Consider the box 
        \begin{equation*}
            \Lambda \coloneqq \{-2^{i+1}n\Delta, \dots, 2^{i+1}n\Delta\}^m.
        \end{equation*}
        Notice that $C_i(\vect)$, $2^iB$, and $C_i(\vect) + 2^iB$ are all contained in $\Lambda$.
        Moreover, $2^iB$ contains only $\vect'$ with $\vect' \equiv \veczero \Mod{2^i}$, while $C_i(\vect)$ and $C_i(\vect) + 2^iB$ only contain $\vect'$ with $\vect' \equiv \vect \Mod{2^i}$. Hence, all three sets are contained in a set of size $\Oh(n\Delta)^m$, and we can compute $C_i(\vect) + 2^iB$ using FFT. The number of operations is $\Oh(n\Delta)^m \cdot \log(n\Delta)$.
    \end{claimproof}

    In total, the algorithm requires one computation of $B$ and $\Oh(i^*) = \Oh(m \log ( \norm{\vecb}_\infty + m\Delta))$ iterations. 
    By \Cref{cl:time}, each iteration requires $\Oh(n\Delta)^m \log(n\Delta)$ operations.
    When $m=1$, running the algorithm from \Cref{thm:packing} requires $\Oh(n\Delta)^m \log(n\Delta)\log(n\Delta\norm{\vecb}_\infty)$ operations.
    This last term also bounds the total number of operations, since
    \begin{equation*}
        \Oh(n\Delta)^m\cdot \log(n\Delta) \cdot \Oh(m \log(\norm{\vecb}_\infty + m\Delta)) \le \Oh(n\Delta)^m\cdot \log(n\Delta) \log(n\Delta\norm{\vecb}_\infty). \qedhere
    \end{equation*}
\end{proof}

Note that one $\log n$ factor only shows up in the case that $m=1$, where our problem immediately reduces to Counting Unbounded Subset Sum with only positive elements.

\paragraph{Running Time Bound.}
We now argue how to transform the bound on the number of arithmetic operations as in \Cref{thm:main} into a bound on the running time in, say, the word RAM model.
To do so, it is sufficient to bound the size of the integers used in (intermediate) operations in the algorithm.
Clearly, any vector $\vecx \in \Z^n_{\ge 0}$ handled by the algorithm is bounded by \[\norm{\vecx}_\infty \le \Oh(\Gamma_{\vecb}) \le \Oh(\norm{\vecb}_\infty m\Delta)^m.\]
Therefore, any intermediate matrix-vector-product $A\vecx$ is still bounded by \[\norm{A\vecx}_\infty \le n \Delta \cdot \norm{\vecx}_\infty \le  n\cdot \Oh(\norm{\vecb}_\infty m\Delta)^m.\]

The number of occurrences of any such vector $A\vecx$ in $B$ is bounded by $|B| = 2^n$. The number of occurrences in $C_i(vecb)$ doubles at most in each iteration. Therefore, the counts handled in the algorithm are bounded by $2^n\cdot \Gamma_{\vecb}$ and the number of bits required tto represent these number is 
\[\log \mleft( 2^n\cdot \Gamma_{\vecb} \mright) \le \Oh(nm)\cdot \log (\norm{\vecb}_\infty m\Delta).\]
In total, this yields a running time of $\Oh(n)^{m+1} \cdot \Oh(\Delta)^m \cdot \log (n\Delta)\cdot \log^2(n\Delta \norm{\vecb}_\infty)$.

\section{Conclusion}

The dependence on $m$ is essentially tight, even in the packing case
(\cref{thm:packing}): Knop, Pilipczuk and Wrochna~\cite{knop2020tight} showed
that even when $A \in \{0,1\}^{m \times n}$, the existence of an algorithm running in time $2^{o(m
\log{m})} \cdot (n\norm{\vecb}_\infty)^{o(m)}$ that decides whether $|\ILP(A,\vecb)| >
0$ would contradict the Exponential Time Hypothesis.  Notice that the dependence on
$\Delta$ is also tight in the packing case.  Namely, for every constant $m \ge
1$ and $\eps > 0$, there is no algorithm running in time
$2^{o(n)} \cdot \Delta^{m-\eps}$, assuming the Strong Exponential Time Hypothesis
(see~\cite{DBLP:conf/stoc/BringmannDW26})\footnote{Bringmann et al.~\cite{DBLP:conf/stoc/BringmannDW26} actually showed that for every $m \ge 1$,
$m$-dimensional Subset Sum is equivalent to Subset Sum, which by~\cite{DBLP:journals/talg/AbboudBHS22} yields SETH-hardness. Both of these results can be adapted
to work for Unbounded Subset Sum variants.}. This was also observed by Jansen
and Rohwedder~\cite{jansen2023integer} in the general case of $A \in
\{-\Delta,\ldots,\Delta\}^{m \times n}$. 

\paragraph{AI Disclosure.}
AI tools (Codex with ChatGPT 5.6) were used for proofreading and suggesting improvements to the exposition. 
The authors assume full responsibility.

\bibliographystyle{abbrv}
\bibliography{bib}

\end{document}